\documentclass[copyright]{eptcs}
\providecommand{\event}{J.~Dassow, G.~Pighizzini, B.~Truthe (Eds.): 11th International Workshop\\
on Descriptional Complexity of Formal Systems (DCFS 2009)} 
\providecommand{\volume}{3}
\providecommand{\anno}{2009}
\providecommand{\firstpage}{117} 
\providecommand{\eid}{11}
\usepackage{breakurl}        

\input{dcfs.tex}

\usepackage{graphics,color,amssymb,multicol}
\usepackage{makeidx}
\usepackage[mathscr]{eucal}

\begin{document}
\newcommand{\w}{,}
\newcommand{\PC}[1]{\footnote{{\bf PIER:} #1}} 
\newcommand{\nul}{\varepsilon}
\newcommand{\empt}{\lambda}
\newcommand{\nO}{{\mathbb N}}
\newcommand{\mN}{\nO_1}
\newcommand{\mNs}[1]{\nO_{#1}}
\newcommand{\natN}{\nO}
\newcommand{\mZ}{{\mathbb Z}}
\newcommand{\mb}[1]{\mathbf{#1}}
\newcommand{\ms}[1]{\mathscr{#1}}
\newcommand{\ra}{\rightarrow}
\newcommand{\RRMsc}{Register machines}
\newcommand{\RM}{register machine }
\newcommand{\RMs}{register machines }
\newcommand{\RMc}{register machine}
\newcommand{\RMsc}{register machines}
\newcommand{\acc}{{acc}}
\newcommand{\sacc}{\mbox{{\tt yes}}}
\newcommand{\icod}{{icod}}
\newcommand{\ncod}{{ncod}}
\newcommand{\srej}{\mbox{{\tt no}}}
\newcommand{\rej}{{rej}}
\newcommand{\inp}{{in}}
\newcommand{\et}{\odot}
\newcommand{\Ra}{\Rightarrow}
\newcommand{\fb}[1]{\fbox{$#1$}}
\newcommand{\nprefix}{{\natN\;}}
\newcommand{\ls}[1]{\mbox{\sf LS}_{#1}}
\newcommand{\sls}[1]{\mbox{\sf SLS}_{#1}}
\newcommand{\fin}{\mbox{\sf FIN}}
\newcommand{\reg}{\mbox{\sf REG}}
\newcommand{\lin}{\mbox{\sf LIN}}
\newcommand{\cf}{\mbox{\sf CF}}
\newcommand{\ol}{\mbox{\sf 0L}}
\newcommand{\eol}{\mbox{\sf E0L}}
\newcommand{\tol}{\mbox{\sf T0L}}
\newcommand{\etol}{\mbox{\sf ET0L}}
\newcommand{\cs}{\mbox{\sf CS}}
\newcommand{\mat}{\mbox{\sf MAT}}
\newcommand{\re}{\mbox{\sf RE}}
\newcommand{\nfin}{\nprefix\mbox{\sf FIN}}
\newcommand{\nreg}{\nprefix\mbox{\sf REG}}
\newcommand{\nlin}{\nprefix\mbox{\sf LIN}}
\newcommand{\ncf}{\nprefix\mbox{\sf CF}}
\newcommand{\ncs}{\nprefix\mbox{\sf CS}}
\newcommand{\nmat}{\nprefix\mbox{\sf MAT}}
\newcommand{\nre}{\nprefix\mbox{\sf RE}}
\newcommand{\nOre}{{\ensuremath\natN_0\mbox{\sf RE}}}
\newcommand{\npre}[1]{{\ensuremath{\natN}_{#1}\;\mbox{\sf RE}}}
\newcommand{\car}[1]{{\ensuremath\nprefix\mbox{\sf CP}_{#1}}}
\newcommand{\symp}[3]{\nprefix\mbox{\sf OP}_{#1}(\mbox{sym}_{#2},\mbox{anti}_{#3})}
\newcommand{\aasymp}[3]{\nprefix\mbox{\sf aOP}_{#1}(\mbox{sym}_{#2},\mbox{anti}_{#3})}
\newcommand{\cat}[2]{\nprefix\mbox{\sf OP}_{#1}(\mbox{cat}_{#2})}
\newcommand{\pcat}[2]{\nprefix\mbox{\sf OP}_{#1}(\mbox{pcat}_{#2})}
\newcommand{\acat}[2]{\nprefix\mbox{\sf aOP}_{#1}(\mbox{cat}_{#2})}
\newcommand{\nacat}[3]{{\natN}_{#1}\;\mbox{\sf aOP}_{#2}(\mbox{cat}_{#3})}
\newcommand{\apcat}[2]{\nprefix\mbox{\sf aOP}_{#1}(\mbox{pcat}_{#2})}
\newcommand{\napcat}[3]{{\natN}_{#1}\;\mbox{\sf aOP}_{#2}(\mbox{pcat}_{#3})}
\newcommand{\ncat}[3]{{\natN}^{#3}\;\mbox{\sf OP}_{#1}(\mbox{cat}_{#2})}
\newcommand{\spl}[6]{\mbox{\sf SPL}_{#1}(#2, #3, #4, #5, #6)}
\newcommand{\splim}[6]{\mbox{\sf SPL}_{im, #1}(#2, #3, #4, #5, #6)}
\newcommand{\lan}[1]{\mbox{\sf L}(#1)}
\newcommand{\num}[1]{\mbox{\sf N}(#1)}
\newcommand{\vecc}[1]{\mbox{\sf W}(#1)}
\newcommand{\vnum}[2]{\mbox{\sf N}^{#1}(#2)}
\newcommand{\bnum}[2]{\mbox{\sf N}_{#1}(#2)}
\newcommand{\vbnum}[3]{\mbox{\sf N}^{#1}_{#2}(#3)}
\newcommand{\bvnum}[3]{\mbox{\sf N}_{#1}^{#2}(#3)}
\newcommand{\splg}[6]{\mbox{\sf SPLG}_{#1}(#2, #3, #4, #5, #6)}
\newcommand{\tspl}[5]{\mbox{\sf tSPL}_{#1}(#2, #3, #4, #5)}
\newcommand{\tsplg}[5]{\mbox{\sf tSPLG}_{#1}(#2, #3, #4, #5)}
\newcommand{\pncat}[3]{{\natN}^{#3}\;\mbox{\sf OP}_{#1}(\mbox{pcat}_{#2})}
\newcommand{\ancat}[3]{{\natN}^{#3}\;\mbox{\sf aOP}_{#1}(\mbox{cat}_{#2})}
\newcommand{\apncat}[3]{{\natN}^{#3}\;\mbox{\sf aOP}_{#1}(\mbox{pcat}_{#2})}
\newcommand{\catc}[2]{\nprefix\mbox{\sf OP}_{#1,-c}(\mbox{cat}_{#2})}
\newcommand{\scatc}[2]{\nprefix\mbox{\sf OP}_{#1,s, -c}(\mbox{cat}_{#2})}
\newcommand{\pcatc}[2]{\nprefix\mbox{\sf OP}_{#1,-c}(\mbox{pcat}_{#2})}
\newcommand{\acatc}[2]{\nprefix\mbox{\sf aOP}_{#1,-c}(\mbox{cat}_{#2})}
\newcommand{\apcatc}[2]{\nprefix\mbox{\sf aOP}_{#1,-c}(\mbox{pcat}_{#2})}
\newcommand{\ncatc}[3]{{\natN}^{#3-#2}\;\mbox{\sf OP}_{#1,-c}(\mbox{cat}_{#2})}
\newcommand{\sncatc}[3]{{\natN}^{#3-#2}\;\mbox{\sf OP}_{#1,s,-c}(\mbox{cat}_{#2})}
\newcommand{\pncatc}[3]{{\natN}^{#3-#2}\;\mbox{\sf OP}_{#1,-c}(\mbox{pcat}_{#2})}
\newcommand{\ancatc}[3]{{\natN}^{#3-#2}\;\mbox{\sf aOP}_{#1,-c}(\mbox{cat}_{#2})}
\newcommand{\apncatc}[3]{{\natN}^{#3-#2}\;\mbox{\sf aOP}_{#1,-c}(\mbox{pcat}_{#2})}
\newcommand{\snp}[9]{\nprefix\mbox{\sf SN}_{#1}^{#2}{\sf P}_{#3}(rule_{#4}, cons_{#5}, gen_{#6}, dlay_{#7}, forg_{#8}, outd_{#9}}
\newcommand{\snpp}[9]{\nprefix\mbox{\sf SN}_{#1}^{#2}{\sf P}_{#3}(rule_{#4}^+, cons_{#5}, gen_{#6}, dlay_{#7}, forg_{#8}, outd_{#9}}
\newcommand{\asnp}[7]{\nprefix\mbox{\sf ASN}{\sf P}_{#1}(rule_{#2}, cons_{#3}, gen_{#4}, forg_{#5}, outd_{#6}, bound_{#7})}
\newcommand{\md}[3]{\nprefix\mbox{\sf OP}_{#1}(\mbox{mate}_{#2},\mbox{drip}_{#3})}
\newcommand{\nmd}[4]{\ensuremath{\natN}_{#1}\;\mbox{\sf OP}_{#2}(\mbox{mate}_{#3},\mbox{drip}_{#4})}
\newcommand{\tp}[3]{\nprefix\mbox{\sf OtP}_{#1}(\mbox{sym}_{#2},\mbox{anti}_{#3})}
\newcommand{\tpp}[3]{\nprefix\mbox{\sf Ot}^\prime{\sf P}_{#1}(\mbox{sym}_{#2},\mbox{anti}_{#3})}
\newcommand{\ssymp}[4]{\nprefix\mbox{\sf O}_{#1}{\sf P}_{#2}(\mbox{sym}_{#3},\mbox{anti}_{#4})}
\newcommand{\assymp}[4]{\nprefix\mbox{\sf aO}_{#1}{\sf P}_{#2}(\mbox{sym}_{#3},\mbox{anti}_{#4})}
\newcommand{\stp}[4]{\nprefix\mbox{\sf O}_{#1}{\sf tP}_{#2}(\mbox{sym}_{#3},\mbox{anti}_{#4})}
\newcommand{\stpp}[4]{\nprefix\mbox{\sf O}_{#1}{\sf t}^\prime{\sf P}_{#2}(\mbox{sym}_{#3},\mbox{anti}_{#4})}
\newcommand{\sps}[4]{#1/a^{#2} \ra a^{#3};#4}
\newcommand{\spf}[1]{a^{#1} \ra \epsilon}
\newcommand{\inn}[3]{(#1, #2, #3)}
\newcommand{\im}[3]{(#1, #2^-, #3)}
\newcommand{\itt}[3]{(#1, #2, #3)}
\newcommand{\ip}[3]{(#1, #2^+, #3)}
\newcommand{\iz}[3]{(#1, #2^{=0}, #3)}
\newcommand{\imm}[4]{(#1, #2^-, #3, #4)}
\newcommand{\imp}[5]{(#1, #2^-, #3^+, #4, #5)}
\newcommand{\impp}[4]{(#1, #2^-, #3^+, #4)}
\newcommand{\izz}[4]{(#1, #2^{=0}, #3, #4)}
\newcommand{\is}[1]{^\bullet#1}
\newcommand{\os}[1]{#1^\bullet}
\newcommand{\an}[2]{(#2; \Out/ #1; \In)}
\newcommand{\oo}[1]{(#1; \Out)}
\newcommand{\ii}[1]{(#1; \In)}
\newcommand{\partI}[1]{{\it Part I:} {\it (#1)}}
\newcommand{\partII}[1]{{\it Part II:} {\it (#1)}}
\newcommand{\St}[1]{\stackrel{#1}{\Rightarrow}}
\newcommand{\Stp}[2]{\stackrel{#1}{\Rightarrow^{#2}}}
\newcommand{\ST}[1]{\stackrel{#1}{\input{./blackArrow.pstex_t}}}
\newcommand{\STp}[2]{\stackrel{#1}{\input{./blackArrow.pstex_t}^{#2}}}
\newcommand{\st}[1]{\stackrel{#1}{\rightarrow}}
\newcommand{\bst}[1]{\stackrel{#1}{{\mathbf \leadsto}}}
\newcommand{\sst}[1]{\stackrel{#1}{\leftrightarrow}}
\newcommand{\stt}[2]{\stackrel{#1}{#2}}
\newcommand{\sympn}[4]{{\natN}_{#1}\;\mbox{\sf OP}_{#2}(\mbox{sym}_{#3},\mbox{anti}_{#4})}
\newcommand{\tpn}[4]{{\natN}_{#1}\;\mbox{\sf OtP}_{#2}(\mbox{sym}_{#3},\mbox{anti}_{#4})}
\newcommand{\psymp}[3]{\nprefix\mbox{\sf PP}_{#1}(\mbox{psym}_{#2},\mbox{anti}_{#3})}
\newcommand{\fsymp}[3]{\nprefix\mbox{\sf PP}_{#1}(\mbox{fsym}_{#2},\mbox{anti}_{#3})}
\newcommand{\ini}{{in}}
\newcommand{\trap}{\star}
\newcommand{\last}{\otimes}
\newcommand{\Out}{\mbox{out}}
\newcommand{\In}{\mbox{in}}
\newcommand{\B}[1]{\mathord{\Bbb #1}}
\newcommand{\NOP}{\B{N}\;\mbox{\sf OP}}
\newcommand{\C}{{\bf PIER }}
\newcommand{\foutje}[1]{\marginpar{\tt\tiny #1}}
\renewcommand{\foutje}[1]{}
\newcommand{\tto}[1]{\stackrel{#1}{\to}}
\newcommand{\ul}[1]{{\underline #1}}
\newcommand{\mono}{1}
\newcommand{\lre}[1]{{\ensuremath #1 \;\mbox{\sf RE}}}
\newcommand{\traces}[4]{#4 \;\mbox{\sf LP}_{#1}(\mbox{sym}_{#2},\mbox{anti}_{#3})}
\newcommand{\rtraces}[3]{\mbox{\sf rLP}_{#1}(\mbox{sym}_{#2},\mbox{anti}_{#3})}
\newcommand{\pt}{P/T $\!^\prime$ }
\newtheorem{prop}{Proposition}
\newtheorem{suggestion}{Suggestion for research}
\newcommand{\rr}[1]{\hspace*{2pt}\framebox{$#1$}\hspace*{2pt}}
\newcommand{\R}[1]{\hspace*{2pt}\framebox{$\rule{0cm}{.3cm}#1$}\hspace*{2pt}}
\newcommand{\fig}[1]{Fig.\ \ref{#1}}
\def\endmarkx{\hskip 2em$\Diamond$\par}
\def\exx{\trivlist \item[]}
\def\endexx{\null\hfill\endmarkx\endtrivlist}

\newcommand{\makeset}[2]{\ensuremath{ \{ #1 \: | \: #2 \} }}
\newcommand{\makesetbig}[2]{\ensuremath{ \big\{ \: #1 \; \big| \; #2 \: \big\} }}
\newcommand{\mymod}[1]{\;(\textup{mod\ }#1)}
\newcommand{\VALC}[1]{\textup{VALC}(#1)}
\renewcommand{\emptyset}{\varnothing}
\renewcommand{\epsilon}{\varepsilon}

\newcommand{\firstmarker}{\dag}
\newcommand{\secondmarker}{\ddag}

\newlength{\cwidth}
\newcommand{\cents}{\settowidth{\cwidth}{c}
\divide\cwidth by2
\advance\cwidth by-.1pt
c\kern-\cwidth
\vrule width .1pt depth.2ex height1.2ex
\kern\cwidth}

\def\sh{\mathbin{\mathchoice
{\rule{.3pt}{1ex}\rule{.3em}{.3pt}\rule{.3pt}{1ex}
\rule{.3em}{.3pt}\rule{.3pt}{1ex}}
{\rule{.3pt}{1ex}\rule{.3em}{.3pt}\rule{.3pt}{1ex}
\rule{.3em}{.3pt}\rule{.3pt}{1ex}}
{\rule{.2pt}{.7ex}\rule{.2em}{.2pt}\rule{.2pt}{.7ex}
\rule{.2em}{.2pt}\rule{.2pt}{.7ex}}
{\rule{.3pt}{1ex}\rule{.3em}{.3pt}\rule{.3pt}{1ex}
\rule{.3em}{.3pt}\rule{.3pt}{1ex}}\mkern2mu}}
\newcommand{\shuff}{\sh}

\title{On Languages Accepted by\\ P/T Systems Composed of {\it joins}}
\def\titlerunning{On Languages Accepted by P/T Systems Composed of {\it joins}}
\def\authorrunning{P.~Frisco, O.\,H.~Ibarra}
\author{Pierluigi Frisco
\institute{School of Mathematical and Computer Sciences -- 
Heriot-Watt University\\
EH14 4AS Edinburgh -- UK}
\email{pier@macs.hw.ac.uk}
\and
Oscar H. Ibarra
\institute{Department of Computer Science --
University of California\\
Santa Barbara -- CA 93106 -- USA}
\email{ibarra@cs.ucsb.edu}
}
\maketitle

\begin{abstract}
Recently, some studies linked the computational power of abstract computing systems based on multiset rewriting to models of Petri nets and the computation power of these nets to their topology.
In turn, the computational power of these abstract computing devices can be understood by just looking at their {\it topology}, that is, information flow. 

Here we continue this line of research introducing {\em J languages} and proving that they can be accepted by place/transition systems whose underlying net is composed only of {\em joins}.
Moreover, we investigate how J languages relate to other families of formal languages.
In particular, we show that every J  language
can be accepted by a  $log~n$ space-bounded
non-deterministic Turing machine with a one-way
read-only input.  We also show that every J
language has a semilinear Parikh map and that J languages and context-free languages (CFLs) are incomparable.
For example, the CFL,  $\{x\#x^R ~|~ x \in \{0,1\}^+\}$,  is
not a J language, but there are non-CFLs that are J languages.
\end{abstract}

\section{Introduction}
\label{sec:intr}
In \cite{Fri05-1} a study on models of Petri nets linking their topological structure to the families of languages they can accept/generate was started.
In particular this study concentrated on Petri nets whose topological structure (that is, their underlying net) was composed only of specific {\it building blocks} ({\it motifs}), that is, little nets connected to each other. 

The following question was raised and partially answered in \cite{Fri05-1}: 
{\it What is the computational power of networks composed of specific building blocks?}
The answer to this question was pursued in \cite{FriTR08,Fbook}.
As shown in \cite{Fri05-1,FriTR08,Fbook} such research can help the study of the computational power of systems based on multiset rewriting.
Given $S_1$, a formal system based on multiset rewriting, the study of its computational power is normally done by proving that it can be simulated by another formal system, say $S_2$, of known computational power. If $S_2$ can also simulate $S_1$, then we can say that the two systems have equivalent computational power.
There is a new way to analyse the computational power of $S_1$ \cite{Fri05-1}.
This new way depends on how the system stores and manipulates
information and it deduces the computational power of $S_1$. 
The way information is stored and manipulated by systems based on multiset rewriting can be easily represented with Petri nets. 
From here then the link between the computational power of formal system based on multiset rewriting and the topological structure of Petri nets.

As indicated in \cite{Fri05-1}, we have not been able to
find in the Petri net literature work that has been done
along the lines of what we  propose.

In the present paper we continue to answer the above question introducing {\em J languages} and proving that they can be accepted by place/transition systems (a model of Petri nets) whose underlying net is composed only of {\em joins} (a kind of building block).
We study how J languages relate to other families of formal languages and show how these relationships allow us to derive the computational power of a model of P systems.

Because of page limit restrictions, several proofs have been omitted.

\section{Basic definitions}
\label{sec:bd}
We assume the reader to have familiarity with basic concepts of formal language theory \cite{HU79}, and in particular with the topic of place/transition systems \cite{RR98,Reis85}.
In this section we recall particular aspects relevant to our presentation.

We denote by $\mN$ the set of natural numbers $\{1, 2, \ldots \}$ while $\nO = \mN \cup \{0\}$.

\begin{definition}
A \index{place/transition system|see{P/T system}}{\em place/transition system} 
(\index{P/T system}{\em P/T system}) is a tuple $$N= (P, T, F,W, K, C_{\mathit{in}})$$ where:
\begin{itemize}
\item[$i)$] $(P, T, F)$ is a {\em net}:
    \begin{enumerate}
    \item $P$ and $T$ are sets with $P \cap T = \emptyset$;
    \item $F \subseteq (P \times T) \cup (T \times P)$;
    \item for every $t \in T$ there exist $p, q \in P$ such that $(p, t), (t, q) \in F$;
    \end{enumerate}
\item[$ii)$] $W: F \rightarrow \mN$ is a \index{function!weight}{\em weight function};
\item[$iii)$] $K: P \ra \mN \cup \{+ \infty\}$ is a \index{function!capacity}{\em capacity function};
\item[$iv)$] $C_{\mathit{in}}: P \rightarrow \nO$ is the \index{P/T system!configuration!initial}{\em initial configuration} (or {\em initial marking}).
\end{itemize}
\end{definition}

We consider P/T systems in which the weight function returns always 1 and the capacity function returns always $+\infty$.
We introduced these functions in the previous definition for consistency with the (for us) standard definition of P/T systems and for consistency with the definition in \cite{Fri05-1,FriTR08,Fbook}.
We follow the very well established notations (places are represented by empty circles, transitions by full rectangle's, tokens by bullets, etc.), concepts and terminology (configuration, input set, output set, sequential configuration graph, etc.) relative to P/T systems \cite{RR98,Reis85}.

In this paper we consider P/T systems as accepting computing devices.
The definition of accepting P/T systems includes the indication of a set $P_{\mathit{in}} \subset P$ of \index{P/T system!input places}{\em input places}, one {\em initial place} $p_{\mathit{init}} \in P\setminus P_{\mathit{in}}$ and one {\em final place} $p_{\mathit{fin}} \in P\setminus P_{\mathit{in}}$.
The places in $P \setminus P_{\mathit{in}}$ are called \index{P/T system!work places}{\em work places}.

An \index{P/T system!accepting}{\em accepting P/T system} $N$ with input $C_{\mathit{in}}$ is 
denoted by 
$$N(C_{\mathit{in}}) = (P, T, F,W, K, P_{\mathit{in}}, p_{\mathit{init}}, p_{\mathit{fin}})$$
where $C_{\mathit{in}}:(P_{\mathit{in}}\cup \{p_{\mathit{init}}\}) \ra \nO,\; C_{\mathit{in}}(p_{\mathit{init}}) = 1$, is the initial configuration of the input places.
So, in the initial configuration some input places can have tokens and the work place $p_{\mathit{init}}$ has one token.
All the remaining places are empty in the initial configuration.
A configuration $C_{\mathit{fin}} \in \mathbb{C}_N$, the set of all reachable configurations of N, is said to be \index{P/T system!accepting!final configuration}{\em final} (or {\em dead state}) if no firing is possible from~$C_{\mathit{fin}}$.

We say that a P/T system  
$N(C_{\mathit{in}}) = (P, T, F, W, K, P_{\mathit{in}}, p_{\mathit{init}}, p_{\mathit{fin}})$ with
$P_{\mathit{in}} =\{p_{in, 1}, \ldots, p_{in, k}\}$,\linebreak
\hbox{$k \in \mN$}, 
\index{P/T system!accepts}{\em accepts} the vector 
$(C_{\mathit{in}}(p_{in, 1}), \ldots, C_{\mathit{in}}(p_{in, k}))$ if in the sequential 
configuration graph of $N(C_{\mathit{in}})$ there is a final configuration 
$C_{\mathit{fin}}$ such that:
\begin{itemize}
\item $C_{\mathit{fin}}(p_{\mathit{fin}}) > 0$;
\item there is at least one path from $C_{\mathit{in}}$ to $C_{\mathit{fin}}$;
\item no other configuration $D$ in the paths from $C_{\mathit{in}}$ to $C_{\mathit{fin}}$ is such that $D(p_{\mathit{fin}}) > 0$.
\end{itemize}

The \index{P/T system!accepting!set of vectors accepted}{\em set of vectors accepted} by $N$ 
is denoted by $\vnum kN$ and it is composed by the vectors 
$$(C_{\mathit{in}}(p_{in, 1}), \ldots, C_{\mathit{in}}(p_{in, k}))$$
accepted by $N$.
The just given definition of (vector) acceptance for P/T systems is new in Petri nets. Normally, the language generated by Petri nets is given by the concatenation of the labels in firing sequences. We discuss this point in Section \ref{sec:fr}. 

As in \cite{FriTR08} we call the nets {\em join} and {\em fork} {\em building blocks}, see Figure \ref{fig:bb}, where the places in each building block are distinct.
\begin{figure}[ht]
\begin{center}
\input{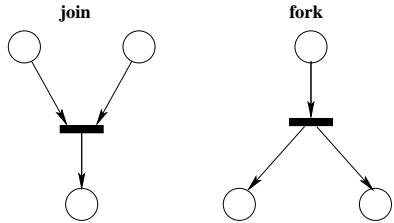}
\end{center}
\caption{Building blocks: {\em join} and {\em fork}.}
\label{fig:bb}
\end{figure}

Also from \cite{FriTR08} we take:
\begin{definition}
Let $x, y \in \{join, fork\}$ be building blocks and let $\bar{t}_x$ and $\hat{t}_y$ be the transitions present in $x$ and $y$ respectively.

We say that {\em y comes after x} (or {\em x is followed by y}, or {\em x comes before y} or {\em x and y are in sequence}) if $\os {\bar{t}_x} \cap {\is {\hat{t}_y}} \neq \emptyset$ and ${\is {\bar{t}_x}} \cap {\is {\hat{t}_y}} = \emptyset$.
We say that {\em x and y are in parallel} if ${\is {\bar{t}_x}} \cap {\is {\hat{t}_y}} \neq \emptyset$ and $\os {\bar{t}_x} \cap {\is {\hat{t}_y}} = \emptyset$.

We say that a net is {\em composed of} building blocks (it is {\em composed of x}) if it can be defined by building blocks (it is defined by {\em x}) sharing places but not transitions.
So, for instance, to say that a net is {\em composed of joins} means that the only building blocks present in the net are {\em join}.
\end{definition}

In this paper we consider accepting P/T systems (in which the weight functions returns always 1 and the capacity function returns always $+\infty$) whose underlying net is composed of joins.
Moreover, if $N = (P, T, F, W, K, P_{\mathit{in}}, p_{\mathit{init}}, p_{\mathit{fin}})$ is such a P/T systems, then for each $t \in T,\; \is t \in (P_{\mathit{in}} \times P\setminus P_{\mathit{in}})$ and $\os t \in P\setminus P_{\mathit{in}}$.
Informally, this means that for each transition $t \in T$ the input set is given by an input place and a work place, while the output set is a work place.
We call these systems {\it J P/T systems}.

It should be clear that J P/T systems are a normal form of accepting P/T systems: for each accepting P/T system there is a J P/T systems accepting the same language.
Such J P/T systems has, eventually, more places and transitions than the original P/T system.
For instance, let us assume that the net depicted in Figure \ref{fig:conv}.a is part of the net underlying an accepting P/T system $N$ with $P$ as set of places, $P_{\mathit{in}} \subset P$ as set input places and $T$ as set of transitions.
The net depicted in Figure \ref{fig:conv}.b belongs to a J P/T system $N_J$ with $P \cup \{w'_1, w'_2\}$ as set of places, $P_{\mathit{in}}$ as set of input places and $T\cup \{t'_1\}$ as set of transitions.
The two nets in Figure \ref{fig:conv} can be regarded as similar in the sets of vectors they accept.
\begin{figure}[ht]
\begin{center}
\input{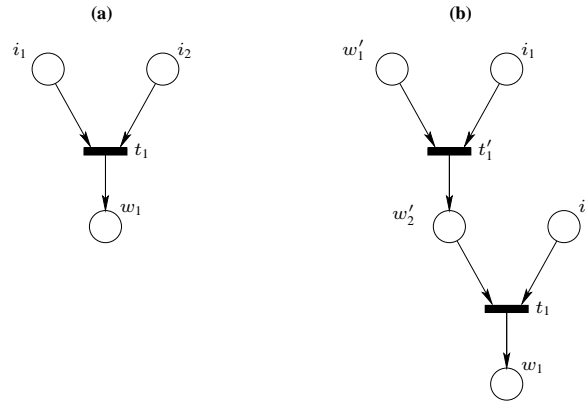}
\end{center}
\caption{{\bf (a)} a net of an accepting P/T system and {\bf (b)} a net of a J P/T system.}
\label{fig:conv}
\end{figure}

\section{J languages and P/T systems}
In this section we prove the main result of the present paper.
In order to do this, we need to introduce a new family of formal languages.

\begin{definition}
\label{def:Jlan}
Let $\Sigma$ be an alphabet, then:
\begin{itemize}
\item $\epsilon$ (the {\em empty string}) is a J expression;
\item for each $v \in \Sigma,\ v$ is a J expression;
\item if $\alpha$ and $\beta$ are J expressions, then $(\alpha \cup \beta)$ is a J expression ({\em union}, in this case $\alpha$ and $\beta$ are called {\it union-terms});
\item if $\alpha$ and $\beta$ are J expressions such that $\alpha, \beta \neq \epsilon$ but they can contain $\epsilon$ (e.\,g., $\alpha = a \cup \epsilon$), then $(\alpha\beta)$ ({\em concatenation}), 
and $(\alpha^+)$ ({\em positive closure}) are J expressions;
\item if $\beta_j,\; 1 \leq j \leq k,\; k \in \mN,$ are J expressions such that none of them contains the operator union and the operator positive closure (the reason for this is explained at page \pageref{pag:upc}), then 
$\beta_1^{n_1} \ldots \beta_k^{n_k}$ ({\em exponentiation} in this case $\beta_j$ are called {\em exponentiation-terms}) is a J expression where each $n_j \in \mN$, called {\em exponent}, is either a fixed positive integer or an integer variable (representing all numbers in $\mN$).
We can specify that some of the exponents are equal.
For example, if $k=8$ it can be that $n_1 = n_3 = n_7 = p,\; n_2 = n_6 = q,\; p, q \in \mN$ ($p$ and $q$ are integer variables), $n_4 = n_8 = 5$ and $n_5 = 3$ ($n_4, n_5$ and $n_8$ are fixed positive integers).
In this case we would have $\beta = \beta_1^p\beta_2^q\beta_3^p\beta_4^5\beta_5^3\beta_6^q\beta_7^p\beta_8^5$.
It is important to note that some of the $\beta_j$s can be $\epsilon$.
\end{itemize}

The language defined by a J expression $\alpha$ is a {\em J language} and it is indicated with $\lan \alpha$.
For instance, $\lan {a \cup \epsilon} = \{a, \epsilon\}$ and $\lan {a^p b^3 a^p} = \{a^p b^3 a^p  |  p \geq 1\}$.

If $\alpha$ is a J expression over the alphabet $\Sigma$, then the {\em length} of $\alpha$ is defined as the number of symbols of $\Sigma\cup\{\epsilon\}$ present in $\alpha$. The length of a J expression is indicated with $\|\alpha\|$.
\end{definition}

The reason why we call these languages {\em J} is because this letter is the initial one in {\em join}, the building block composing the nets considered in this paper.

In writing J expressions we can omit many parentheses is we assume that positive 
closure and exponentiation have precedence over concatenation or
union, and that concatenation has precedence over union.
So, for instance, it is possible to write J expressions as 
$\alpha = \epsilon \cup (ab^+\cup b)^+ \cup a^p(bc)^qc^3a^pb^2(cd)^q$.

\begin{remark}
\label{rem:one}
If $\beta$ is an exponentiation with fixed positive integer exponents,
we can construct another exponentiation $\beta'$ such that
$\lan \beta = \lan \beta'$ and $\beta'$ has fixed
positive integer constants that are all 1's.
\end{remark}

The previous remark is clearly true: for each $\beta_k$ exponentiation-term in $\beta$ having $n_k$ as fixed positive exponent, $\beta'$ can be obtained concatenating $n_k$ times~$\beta_k$.
So, for instance, if $\beta = a^p(bc)^qc^3a^pb^2(cd)^q$, then $\beta' = a^p(bc)^qccca^pbb(cd)^q$. 

If $\Sigma$ is a set, then $|\Sigma|$ denotes the {\em cardinality} of $\Sigma$, that is the number of elements in $\Sigma$.
The following follows from Definition \ref{def:Jlan}:
\begin{lemma}
Let $\beta$ be an exponentiation-term. Then:
\begin{itemize}
\item if $\epsilon \in \lan \beta$, then $\lan \beta = \{\epsilon\}$;
\item if $|\lan \beta| > 1$, then $\epsilon \not\in\lan \beta$.
\end{itemize}
\end{lemma}

The proof of the following lemma is rather long but not particularly difficult.
The basic idea is to have a J P/T system in which input places are associated to the J expression defining the language accepted by the J P/T system, work places are associated with the possible union, concatenations, positive closure and exponentiations of the J language.
The J P/T system repeatedly ``consumes" (accepts) one token per time from the input places and passes one token from a work place to another.
The J P/T system is non-deterministic (because it ``guesses" to what part of the J expression a token can be matched).

\begin{lemma}
\label{lem:J}
Every J language is accepted by a J P/T system.
\end{lemma}

\label{pag:upc}
Before presenting the next results we explain why exponentiation-terms have to be different than union and positive closure.
Let $\beta = \beta_1^{n_1}\beta_2^{n_2}\beta_3^{n_1}\beta_4^{n_2}$ be an exponentiation with $\beta_1^{n_1}, \beta_2^{n_2}, \beta_3^{n_1}, \beta_4^{n_2}$ exponentiation terms. 
There is no meaning in having (for instance) $\beta_1 = \alpha^+$, where $\alpha$ is a J expression, as $\beta_1^{n_1} = \alpha^{+^{n_1}} = \alpha^+$. 
So, $\beta = \alpha^+\beta_2^{n_2}\beta_3^{n_1}\beta_4^{n_2}$ is the concatenation of $\alpha^+$ to an exponentiation.
A similar argument holds if an exponentiation term contains a positive closure, that is, for instance, $\beta = \alpha\gamma^+$ where $\alpha$ and $\gamma$ are J expressions.

The reason why exponentiation-terms cannot be union depends on the fact that J P/T systems do not have memory.
Let $\beta$ be defined as in the above, let $n_2 > 1$ and let (for example) $\beta_2 = \alpha_1 \cup \alpha_2$, where $\alpha_1$ and $\alpha_2$ are J expressions.
This means that $\beta = \beta_1^{n_1}(\alpha_1 \cup \alpha_2)^{n_2}\beta_3^{n_1}\beta_4^{n_2}$. 
Let us assume that in the initial configuration of the J P/T system accepting $\beta$ there are some tokens in the input places associated to $\alpha_1$ and to $\alpha_2$. 
We know from Lemma \ref{lem:J} that the check of the presence of symbols in $\beta_2$ and $\beta_4$ is done in passages: first checking the occurrence of symbols in $\beta_2$, then the one in $\beta_4$, then (second passage) the one in $\beta_2$ again, and so on.
It can be that (as the J P/T system does not have memory) in the first passage tokens related to $\alpha_1$ are checked, while in the second passage tokens related to $\alpha_2$ are checked.
This would not be a desired behaviour.

The fact that exponentiation-terms cannot be union is not a big limit as we can 
rewrite $\beta$ as 
$$\beta_1^{n_1}\alpha_1^{n_2}\beta_3^{n_1}\beta_4^{n_2} 
\cup \beta_1^{n_1}\alpha_2^{n_2}\beta_3^{n_1}\beta_4^{n_2}.$$

Here a concept that we need in the following:
\begin{definition}
Let $N$ be a J P/T system.
We say that $N$ contains cycles if and only if some firing sequences of $N$ are of the kind $\alpha\beta^n\gamma \in T^*$, where $T$ is the set of transitions of $N$ and $n > 1$.
A {\em cycle} is a cyclic path in the net underlying $N$ having $\beta$ as sequential transitions in a firing sequence.

We denote cycles with the sequence of pairs of places and transitions belonging to it.
The {\em length of a cycle} is the number of transitions present into it.
\end{definition}

Here the converse of the previous lemma:
\begin{lemma}
Every language accepted by a J P/T system is a J language.
\end{lemma}
\begin{proof}
We only provide a sketch of the proof a detailed proof would be tedious.
It is very important to recall that:
\begin{itemize}
\item the underlying topological structure of J P/T systems is composed by {\em join} and that for each transition the input set is given by an input place and a work place;
\item the initial configuration sees tokens in input places and in only one work place (the initial place).
\end{itemize}

Let $N$ be a J P/T system and let its input places be associated to symbols in an alphabet $\Sigma$. 
If $N$ contains no cycle, then $N$ accepts concatenations of symbols and unions of symbols and their concatenation.
If instead $N$ contains cycles, then this means that concatenations of symbols can be repeatedly checked.
This means that $N$ can accept the positive closure of symbols, concatenations and their union.

Now we prove that $N$ can accept exponentiations.
Let us assume that $N$ accepts $\beta_1^+\beta_2^+$ with 
$\beta_1 = \beta_{1, 1}\beta_{1, 2} \ldots \beta_{1, k_1}$, 
$\beta_2 = \beta_{2, 1}\beta_{2, 2} \ldots \beta_{2, k_2}$,
$\beta_{1, i}, \beta_{2, j} \in \Sigma^+$, 
$1 \leq i \leq k_1$, $1 \leq j \leq k_2$. 
In order to simplify the proof we assume that $k_1 = k_2$.
With slight modifications the result holds also if $k_1 \neq k_2$.

It is possible to define another J P/T system $N'$ accepting $\beta_{1, 1}^{n_1}\hfill \beta_{1, 2}^{n_2}\hfill \ldots\hfill \beta_{1, k_1}^{n_{k_1}}$\\$\beta_{2, 1}^{n_1}\beta_{2, 2}^{n_2}\ldots\beta_{2, k_1}^{n_{k_1}}$.
The system $N'$ is very similar to $N$.
It is made such that when the last symbol of $\beta_{1, 1}$ is checked, then the first symbols of $\beta_{2, 1}$ is checked.
When the last symbol of $\beta_{2, 1}$ is checked, then the system can either check the first symbol of $\beta_{1, 1}$ or the first symbol of $\beta_{1, 2}$ and so on.
The same result holds if either $\beta_1$ or $\beta_2$ is not a positive closure (but just a concatenation).
Informally: for J P/T systems exponentiation is a shuffling of concatenations.

Now we prove that nothing else can be accepted by J P/T systems.
By contradiction, let us assume that there is a set of vectors accepted by a J P/T system having $P_{\mathit{in}}$ as set of initial places such that it cannot be represented by a J expression over $P_{\mathit{in}}$.
Clearly, the set of vectors has to have an infinite number of elementents.
If not, then a J expression given by the union of the concatenations of the different elements in each of the finite number of vectors would represent this set.

As the number of places and transitions is finite, then the number of cycles in the J P/T system is finite, too.
Depending on the number and the length of the cycles present in the J P/T system, there is a finite set of accepted initial configurations (called {\em border configuration}) such that for each of them there are vectors (called {\em added vector}) such that the (vector) sum of one border configuration to any multiple of any of its added vector leads to an accepted initial configuration.
Informally, the acceptance of any border configuration needs some cycles to be traversed. Given a border configuration, its added vectors allow these cycles to be traversed other times. 
But then, there is a J expression that can represent the set of vectors accepted by the J P/T system.
This J expression is given by the union of J expressions representing border configurations where each place is concatenated with the respective place in the added vectors to the power of an integer variable.
A contradiction.

For instance, let $P_{\mathit{in}} = \{p_1, p_2\}$, $(4, 6)$ be a border configuration, and let $(2, 0)$ and $(1, 3)$ be added vectors for the border configuration.
The J expression is then: $p_1^4(p_1p_1)^{k_1}p_2^6 \cup p_1^4p_1^{k_2}p_2^6(p_2p_2p_2)^{k_3}$ where $k_1, k_2, k_3 \in \mN$ are integer variables.%
\end{proof}

From the previous two lemmas we have:
\begin{theorem}
\label{th:xpt}
A language is a J language if and only if it is accepted by a J P/T system.
\end{theorem}

\section{Semilinearity of J languages}
In this section, we show that the Parikh map of every J languages
is semilinear.  
We also prove a ``converse'' (this is made more precise later)
of this result.

Let $N$ be the set of non-negative integers and $n$ be a positive
integer.  A subset $S$ of $N^n$ is a {\em linear set} if there
exist vectors $v_0,v_1, \dots , v_t$ in $N^n$ such that
$$S = \{ v \mid v = v_0 + i_1v_1 + \cdots+ i_tv_t, \
i_j \in N\}.$$
The vectors $v_0$ (referred to as the {\em constant vector}) and
$v_1, v_2, \dots , v_t$ (referred to as the {\em periods}) are called
the {\em generators} of the linear set $S$.
The set $S \subseteq N^n$ is {\em semilinear} if
it is a finite union of linear sets.

The empty set is a trivial (semi)linear set, where the set of
generators is empty.
Every finite subset of $N^n$ is
semilinear -- it is a finite union of linear sets whose generators
are constant vectors. It is also clear that the semilinear sets
are closed under (finite) union.

Let $\Sigma = \{a_1, a_2, \dots , a_n\}$ be an alphabet.
For each word $w$ in $\Sigma^*$, define the
Parikh map of $w$ to be
$$\psi(w) = (|w|_{a_1}, |w|_{a_2}, \dots , |w|_{a_n}).$$

\noindent
where $|w|_{a_i}$ denotes the number of occurrences of symbol $a_i$ in $w$.
For a language $L \subseteq \Sigma^*$, the Parikh map of $L$ is
$\psi(L) = \{ \psi(w) \mid w \in L\}$.
The language $L$ is semilinear if $\psi(L)$ is a semilinear set.

There is a simple automata characterisation of semilinear sets.
Let $M$ be a non-deterministic finite automaton
{\em without an input tape}, but with $n$ counters
(for some $n \ge 1$). The computation of $M$ starts
with all the counters zero and the automaton in the start state.
An atomic move of $M$ consists of incrementing at most one
counter by 1 and changing the state (decrements are not allowed).
An $n$-tuple $v =(i_1, \dots,i_n) \in N^n$ is generated by $M$ if $M$,
when started from its initial configuration, halts with
$v$ as the contents of the counters.
The set of all $n$-tuples generated by $M$ is denoted by $G(M)$.
We call this automaton a {\em finite-state generator}.

The following result was shown in \cite{HIKS02}:

\begin{theorem} \label{gen}
Let $n \geq 1$.
A subset $S \subseteq N^n$ is semilinear
if and only if it can be generated
by a finite-state generator with $n$ counters.
\end{theorem}

Using Theorem \ref{gen}, we can then prove the following result.

\begin{theorem} \label{semi}
The Parikh map of every language denoted by a J expression is semilinear.
\end{theorem}

\noindent
For the ``converse'' of Theorem \ref{semi}, we need the following definition.

\begin{definition}
\label{def:Lset}
Let $S \subseteq N^n$ and
$\Sigma = \{a_1, \dots, a_n\}$.
Define the language 
$$L_S = \{a_1^{s_1} a_2^{s_2} \cdots a_n^{s_n} \mid (s_1, \dots, s_n) \in S \}.$$
\end{definition}

\begin{theorem} \label{semi-converse}
If $S$ is a semilinear set, then $L_S$ is a J language.
\end{theorem}

\section{Complexity of J Languages}

Here, we briefly discuss the (TM) space complexity of J languages.
We will show that every J language can be accepted by a
non-deterministic Turing machine (NTM)
with a one-way read-only input and
a $log~n$ space-bounded read-write work-tape.
Actually, what we show is that the language can be accepted by
a one-way non-deterministic finite automaton augmented with a finite
number of counters.  In each computing step each counter can be incremented/decremented by
1 and tested for zero.  The counters start with zero value,
and we assume (without loss of generality) that
the machine accepts when in the final state and when all counters store zero.
During the computation,
the (non-negative) integer value in each counter never
exceeds the length of the one-way read-only input.
We call this machine a linear-space multicounter machine,
or simply, LCM.  Clearly, an LCM can be simulated by
a one-way $log~n$ space-bounded NTM, since the values in the
counters can be stored and managed on a $log~n$ read-write work-tape.

The next two results can be shown.

\begin{theorem}
Every J language can be accepted by an LCM.
\end{theorem}

\begin{corollary}
Every J language can be accepted by a one-way $log~n$
space-bounded NTM.
\end{corollary} 

It is well-known and, actually easily shown, that
$L = \{x\#x^R ~|~ x \in \{0,1\}^+ \}$
($R$ denotes reverse) cannot be accepted by a one-way $log~n$
space-bounded NTM, hence, cannot be accepted by an LCM.
(For an input $x\#x^R$ of length $2n+1$, a one-way NTM
with $log~n$ space can only differentiate a linear
number of strings of $x$'s before the symbol $\#$.
But there are $2^n$ different $x$'s.)

\begin{corollary}
There are context-free languages that are not J languages.
\end{corollary}

\section{A grammatical characterisation of J languages}

In this section, we provide a grammatical characterisation of
J languages. The grammar is an extension of the right-linear simple
matrix grammar studied in~\cite{Ib70}.
  
Let $\Sigma$ be the set of terminal symbols.
The non-terminal symbols are partitioned into
two disjoint sets, $\cal Q$ and $\cal R$.
There is a unique start non-terminal $S_0 \in \cal Q$
from which all derivations start from.
The rules are of two types:\\

\noindent
{\bf Basic Rules:}

\begin{enumerate}

\item

$S \rightarrow  w$,  where $w \in \Sigma \cup \{\epsilon\}$
and $S \in \cal Q$ does not appear on the RHS of any  basic
rule, but can appear in a matrix rule 6 below.

\item
$S \rightarrow  S_1 | S_2$, where $S, S_1, S_2$ are
distinct non-terminals in $\cal Q$, and $S$ does not
appear on the RHS of any basic rule, but can appear in a matrix
rule~6 below.

\item

$S \rightarrow  S_1  S_2$,  where $S, S_1, S_k$ are
distinct non-terminals in $\cal Q$, and $S$ does does not appear
on the RHS of any basic rule, but can appear in a matrix rule 6 below.

\item

$S \rightarrow SS$, where $S \in \cal Q$
does not appear on the RHS of any basic rule (except in this rule),
but can appear in a matrix rule 6 below.

\item

$S \rightarrow (A_{11} A_{12} \cdots A_{1m}, \dots, A_{k1} A_{k2} \cdots A_{km})$,
where $m \ge 1$, $k \ge 1$, 
each $A_{ij}$ is a non-terminal in $\cal R$ and
$S \in \cal Q$ can appear on the RHS of basic rules~2, 3, 4,
but cannot appear in a matrix rule 6 below.
\end{enumerate}

\noindent
{\bf Right-Linear Simple Matrix Rules:}

\begin{enumerate}
\item [6.]
$[A_1 \rightarrow S_1A_1 , \dots,  A_k \rightarrow S_kA_k]$,
where $k \ge 1$, each $A_i$ a non-terminal in $\cal R$,
and each $S_i \in \cal Q$
(subject to the restriction in rule 5 above).

{\em Restriction 1:} We require that if
$[A_1 \rightarrow S_1A_1 , \dots,  A_k \rightarrow S_kA_k]$ and
$[A_1 \rightarrow S_1'A_1 , \dots,  A_k \rightarrow S_k'A_k]$
are both matrix rules, then $S_i = S_i'$ for $1 \le i \le k$.
Thus, the RHS is unique for the given $A_i$'s on the LHS.

\item [7.]
$[A_1 \rightarrow w_1, \dots, A_k \rightarrow w_k]$,
where $k \ge 1$, each $A_i$ a non-terminal in $\cal R$,
each $w_i$ in $\Sigma^*$.
\end{enumerate}

\noindent
The derivation of a string $w \in \Sigma^*$ in the language starts from the
non-terminal~$S_0$.  If at some point during the derivation,
an intermediate string is reached that contains a non-terminal $S$
for which a rule of form~5 is applied, this $S$
will be replaced by an $n$-tuple
$(A_{11} A_{12} \cdots A_{1m}, \dots, A_{k1} A_{k2} \cdots A_{km})$.
Next, a rule of form 6 is applied in parallel, i.\,e., 
application of the rule rewrites the leftmost
non-terminal of each of the $k$-coordinates.  Application
of rule 6 is done $r \ge  0$ times, where $r$ is chosen
non-deterministically; after which rule 7 is applied.
The process is repeated
for the next leftmost non-terminal of each coordinate.
At the end, when all $k$ coordinates are non-null
strings in $\cal Q^+$, we
``merge'' the $k$ components into a single string.
Then the derivation continues until $w$ is reached.

We can prove the following result.

\begin{theorem} \label{thm1}
The languages generated by  ERLSMGs
are exactly the J languages, which allow union
and positive closure in exponentiation.
\end{theorem}

\begin{corollary} \label{cor1}
The languages generated by ERLSMG's in which the S's on the left-hand-side of rules
of forms 2 and 4 do not appear on the right-hand-sides of rules of form 6 are exactly the J languages.
\end{corollary}

\section{Final remarks}
\label{sec:fr}
In Section \ref{sec:bd} we said that the way to accept languages (sets of vectors) considered by us differs from the standard one used in Petri nets (concatenations of the labels of firing sequences) \cite{hackBook,Jan86}.
The reason why we did not consider this standard way in the present paper is because we wanted here to focus only on the topology. 
(We are in the process of writing a paper discussing the relations
between these two different ways of accepting languages).

In \cite{Fbook,FriTR08} it is shown how the results obtained from the computational power of P/T system whose underlying net is composed of {\em joins} and {\em fork} can facilitate the study of the computational power of models of {\em membrane systems} (also known as {\em P systems}) \cite{GPa00} based on multiset rewriting. 
 These results use a definition of {\it equivalence} (also present in \cite{Fbook,FriTR08}).
This is the ``new way to analyse the computational power of a formal system" we mentioned in Section \ref{sec:intr}.

In a nutshell, the idea is the following: if a formal system $S$ can simulate {\em fork}, {\em join} and their composition, then the results on the computational power of P/T systems whose underlying net is composed of {\em joins} and {\em fork} are also valid to $S$.

In \cite{Fbook,FriTR08} it is shown that P systems with catalysts can simulate a {\em fork} using rules of the kind $a \ra b_1b_2$, while the simulation of a {\em join} does not require the use of such rules.
So, knowing from \cite{Fbook,FriTR08} how P systems with catalysts can simulate {\em join} and Theorem \ref{th:xpt}, we can say that the family of languages generated by P systems with catalysts not using rules of the kind $a \ra b_1b_2$ is J.

Using the definitions and results of P systems with catalysts in \cite{Fbook,FriTR08} we can be more precise and state:

\begin{corollary}
\mbox{ }
\begin{itemize}
\item The family of languages accepted by P systems with catalysts of degree 2 and 2 catalysts not using rules of the kind $a \ra b_1b_2$ is J;
\item the family of languages accepted by purely catalytic P systems of degree 2 and 3 catalysts not using rules of the kind $a \ra b_1b_2$ is J.
\end{itemize}
\end{corollary}

We end this paper with an open problem.

In the rule of form 6, we had a restriction that if
$$[A_1 \rightarrow S_1A_1 , \dots,  A_k \rightarrow S_kA_k] \mbox{ and }
[A_1 \rightarrow S_1'A_1 , \dots,  A_k \rightarrow S_k'A_k]$$
are both matrix rules, then $S_i = S_i'$ for $1 \le i \le k$.
Suppose we remove this restriction. 
Is there an extension of the J P/T systems that can characterise these grammars?

\bibliographystyle{eptcs}
\bibliography{frisco}

\end{document}